\documentclass[12pt]{article}

\usepackage{amsmath}
\usepackage{amssymb}     % \complement symbol
\usepackage{fullpage}    % Use normal margins
\usepackage{float}       % [H] float specifier
\usepackage{enumitem}    % Enumerate environment
\usepackage{mathpazo}    % Normal font
\usepackage{inconsolata} % Teletype font
\usepackage{hhline}      % \hhline tabular delimiter

\usepackage[numbers]{natbib}
\newcommand{\textcite}{\citet} % Forward compatibility to Biblatex
\newcommand{\autocite}{\citep}

\usepackage{amsthm}      % Theorem, corollary, & lemma environments
\newtheorem{theorem}{Theorem}
\newtheorem{corollary}{Corollary}[theorem]
\newtheorem{lemma}{Lemma}

\usepackage{hyperref}    % Hyperlinks
\hypersetup{
	pdfinfo={
		Title={The Broadcaster Repacking Problem},
		Author={William K.~Schwartz},
	}
}

\newcommand{\N}{\mathbb{N}} % Natural numbers
\DeclareMathOperator{\BRPF}{BRPF}
\DeclareMathOperator{\BRP}{BRP}
\newcommand{\introduce}[1]{\textbf{\textit{#1}}} % For defining new terms

\begin{document}
\title{The Broadcaster-Repacking Problem}
\author{
	William K. Schwartz\thanks{
		Department of Applied Mathematics, Illinois Institute of Technology, Chicago, IL 60616 \hfill\break(\href{mailto:wschwart@hawk.iit.edu}{\nolinkurl{wschwart@hawk.iit.edu}}).
	}
}
\maketitle

\begin{abstract}
	The Federal Communications Commission's (FCC's) ongoing Incentive Auction will, if successful, transfer billions of dollars of radio spectrum from television broadcasters to mobile-network operators.
	Hundreds of broadcasters may go off the air.
	Most of those remaining on the air, including hundreds of Canadian broadcasters not bidding, will have to move to new channels to continue broadcasting.
	The auction can only end if all these broadcasters will fit into the spectrum remaining for television.
	Whether a given set of broadcasters fits is the broadcaster-repacking problem.
	The FCC must calculate its solutions thousands of times per round of bidding.
	Speed is essential.

	By reducing the broadcaster-repacking problem to the maximum independent set problem, we show that the former is $\mathcal{NP}$-complete.
	This reduction also allows us to expand on sparsity-exploiting heuristics in the literature, which have made the FCC's repacking-problem instances tractable.
	We conclude by relating the heuristics to satisfiability and integer programming reductions.
	These provide a basis for implementing algorithms in off-the-shelf software to solve the broadcaster-repacking problem.
\end{abstract}

\section{Introduction}

	In the ongoing Broadcaster Incentive Auction \autocite{fcc-stages}, the Federal Communications Commission (\introduce{FCC}) is trying to buy radio spectrum from 1,877 broadcast television stations \autocite{fcc-reverse-opening-prices} to sell to 62 mobile-network operators, including Verizon, AT\&T, and T-Mobile \autocite{fcc-forward-auction-qualified-bidders}.
	The Incentive Auction comprises pairs of alternating sub-auctions occurring in \introduce{stages}.
	In each stage, the FCC sets the \introduce{clearing target} in megahertz of how much spectrum the FCC aims to transfer.
	The first sub-auction, the \introduce{Reverse Auction}, allows broadcasters to bid for the right sell spectrum to the FCC.
	The FCC accepts enough bids to satisfy this stage's clearing target.
	Then the second sub-auction, the \introduce{Forward Auction}, allows the networks to bid for the frequencies that the stage's Reverse Auction made available.
	After a pair of Reverse and Forward Auctions do not clear the market, the FCC lowers the clearing target and starts a new stage.
	At the time of this writing, the first stage did not clear the market and the second stage is halfway done.
	In the first stage, the broadcasters agreed to sell 120 MHz at \$86.4 billion, but the networks agreed to buy 120 MHz at only \$23 billion.
	In the second stage, the broadcasters agreed to sell 114 MHz at \$54.6 billion \autocite{Gibbs:2016ti}.
	The second stage's Forward Auction has not yet completed \autocite{fcc-auction-dashboard}.
	Once a stage completes with the Forward Auction's revenue exceeding the Reverse Auction's costs (and some other complicated rules), the Incentive Auction is over.
	The bids in the final stage are binding.

	Winners of the Reverse Auction sell their broadcast licenses back to the FCC and then go out of business.\footnotemark
	\footnotetext{
		Broadcasters can also win money in the Reverse Auction not to go off the air but to share a channel with another TV station.
		We ignore this complication.
	}
	Losers of the Reverse Auction sell no frequencies and continue broadcasting after the Incentive Auction closes, but on new channels.
	The government chooses the new channels so the losing broadcasters occupy the low end of the spectrum, leaving the high end for the mobile networks.
	The FCC must assign the channels so that the stations do not interfere with one another's signals.
	This reassigning process is called \introduce{repacking}.
	To ensure that the Reverse Auction can end with all the losing bidders being able to continue broadcasting, the FCC must check at each step through the Reverse Auction that such a repacking exists.
	Determining whether a feasible repacking exists is the \introduce{broadcaster-repacking problem} (\introduce{BRP}).
	Because of the size and structure of the Reverse Auction, the FCC needs to calculate thousands of BRP instances very quickly \autocite{Frechette:2016it}.
	Our concern in this paper is repacking.
	We will not discuss the auction structure in more detail.

	In repacking, the FCC obeys three sets of constraints: the domain constraints, listing which channels a station may occupy; the interference constraints, listing which stations cannot be on certain channels at the same time; and the at-most constraints, listing that each station may occupy at most one channel.
	(Other software determines what these constraints are. See \autocite{fcc-tvstudy}.)
	On its face, this problem is similar to graph coloring because it involves making assignments to avoid given sets of pairwise relationships.
	As we will show, it is more convenient to solve the BRP by thinking of it as a maximum independent set problem.
	From a BRP instance, we will construct a graph whose edges are the interference constraints and whose vertices are not stations, but station-channel pairs.
	Solving the maximum independent set problem on this constraint graph is equivalent to the BRP, and allows some heuristic simplifications of the problem that are not obvious from reductions to other $\mathcal{NP}$-complete problems.

	The FCC is using a satisfiability-cum-local-search solver called SATFC designed by a team at University of British Columbia.
	The software is freely available \autocite{satfc}.
	This author also implemented a satisfiability-based solver (with somewhat different design requirements), named TVRepack, starting in late August, 2013 \autocite{tvrepack}.

	\subsection{Prior work}

		To our knowledge, the SATFC team, led by Kevin Leyton-Brown, Neil Newman, and Alexandre Fr\'e\-chette, has published the only academic paper on the BRP \autocite{Frechette:2016it}.
		It gives only a minimal mathematical description of the BRP, instead focusing on the authors' specialized solver configuration and novel caching techniques.

		The FCC's auction regulations give a mathematical description of the BRP as a zero-one integer program.
		This description is unambiguous enough for the legal and political goals of ensuring all auction participants understand how repacking will work, but the notation is a tad clumsy for proving general theorems.
		The novelty of this paper is the notation, the rigorous proof of the BRP's $\mathcal{NP}$-completeness, the definition and use of the ``constraint graph'', and an extension of decomposition heuristics of \textcite{Frechette:2016it}.

	\subsection{Outline}

		The remainder of the paper proceeds as follows:

		Section~\ref{model} details the mathematical model for the BRP.
		The effort has two goals.
		First, it provides the language that later sections will use for deriving theorems about the BRP.
		Second, it introduces concepts important for understanding the FCC's data set and some of its special structure.

		Section~\ref{brp-as-a-graph-problem} reformulates the BRP as a graph problem by constructing a graph whose maximum independent sets are solutions to the BRP.
		It then formally proves $\mathcal{NP}$-completeness.

		Section~\ref{decomposition} generalizes the heuristics that \textcite{Frechette:2016it} proposed.
		It also explores the FCC's BRP constraints data set to gauge how these decomposition heuristics affect the FCC's problem instances.

		Section~\ref{sec:zoip} provides reductions of the BRP to satisfiability and integer programming as complements to the independent-set approach.
		These are formulations that off-the-shelf optimization libraries can solve directly.

		Section~\ref{sec:future-research} concludes with a discussion of future research possibilities.

\section{Model and Notation}\label{model}

	The broadcaster-repacking problem is a generalization of list coloring.
	In \introduce{list coloring}, we are given a simple graph $G$ (an undirected graph with neither loops nor multiple edges) and \introduce{lists} $L(v)$ for each vertex $v$.
	The lists are sets of \introduce{colors}, typically natural numbers, with which we are allowed to paint the corresponding vertices.
	The problem is to paint each vertex using colors only from the corresponding list while never painting neighboring vertices the same color.

	Broadcaster repacking differs from list coloring superficially in terminology: vertices become television stations, colors become channels, lists become domains, and painting becomes assigning.
	More substantively, the BRP relaxes list coloring's rule forbidding neighbors from being painted the same color by specifying different interference constraints between stations depending on the channel.
	A BRP instance might stipulate that station WABC and station KXYZ cannot both broadcast on channel five but they may concurrently occupy channel nine.
	Another wrinkle comes from offset constraints: WABC cannot be on channel six while KXYZ is on channel seven.
	In the FCC's Incentive Auction, the physics of broadcasting on radio frequencies and laws protecting the economic interests of established broadcasters necessitate these complicated interference patterns.
	Consequently, list coloring is not quite an expressive enough model for the BRP without a little more work.
	The following subsections undertake that work.

	\subsection{Constraints}\label{sec:constraints}

		Let the sets \(S, C \subset \mathbb{N}\) be finite sets of \introduce{stations} and \introduce{channels} respectively.
		For some \(S' \subseteq S\), an \introduce{assignment} \(f : S' \to C\) maps each station \(s \in S'\) to a channel \(f(s) \in C\).
		The assignment is \introduce{complete} if \(S' = S\), that is, if \(f\) is defined for every station.
		Finally, an assignment \(f : S' \to C\) is \introduce{feasible} if it satisfies the following three constraints.

		\begin{enumerate}
			\def\labelenumi{\arabic{enumi}.}
			\item
				The \introduce{interference constraints} are elements of some given set \(I\) where
				\[
					I \subseteq (S \times C)^2 = \left\{\left((s_1, c_1), (s_2, c_2)\right) \mid s_1,s_2 \in S \text{ and } c_1,c_2 \in C\right\}.
				\]
				They indicate that, if the \introduce{protected station} \(s_1\) is assigned to channel \(c_1\), then the \introduce{interfering station} \(s_2\) cannot be assigned to channel \(c_2\).
				More precisely, \(f\) is feasible only if \(f(s_1) = c_1 \implies f(s_2) \ne c_2\) for all \(((s_1, c_1), (s_2, c_2)) \in I\).
			\item
				The \introduce{domain constraints} are some given sets \(C_s \subseteq C\) for each \(s \in S\).
				They indicate that station \(s\) may be assigned only channels in \(C_s\).
				For convenience we bundle all domain constraints into a set \(D = \left\{C_s\right\}_{s \in S}\).
				More precisely, \(f\) is feasible only if \(f(s) \in C_s\) for all \(s \in S\).
			\item
				The \introduce{at-most constraints} require that each station be assigned at most one channel.
				This is redundant with the definition of \(f\) as a function, but having a name for it will be convenient when we consider subsets of \(S \times C\) later.
		\end{enumerate}

	\subsection{Simplifying assumptions}\label{simplifying-assumptions}

		Without loss of generality, we may assume that \(C = \bigcup_{s \in S}C_s\) and that the sets of stations represented in \(D\) and in \(I\) both equal \(S\).
		With these assumptions, the \introduce{input data} of a problem instance are
		\begin{enumerate}
			\item the set \(I\) of interference constraints, and
			\item the set \(D\) of domain constraints.
		\end{enumerate}
		The BRP instances the FCC needs to solve in the Reverse Auction are subproblems of the BRP instance whose input data are the interference and domain constraints \((I, D)\) found in the FCC's published constraint files \autocite{fcc-constraint-files}.

		\(I\) is actually much smaller than \((S \times C)^2\).
		We assume that \(I\) has no interference constraints of the form \(((s, c_1), (s, c_2))\), as these would be redundant with the at-most constraints.
		Further, we assume that there exists a finite set of \introduce{offsets} \(\mathcal{O} \subset \mathbb{Z}\) so that every element of \(I\) is of the form \(\left((s_1, c), (s_2, c + k)\right)\) for some \(s_1, s_2 \in S\), \(c \in C\), and \(k \in \mathcal{O}\).
		When $k = 0$, the FCC calls such an interference constraint a \introduce{co-channel} constraint; otherwise, it's an \introduce{adjacent $k$} constraint.
		The FCC's data has $\mathcal{O} = \{-2, -1, \ldots, 2\}$, and whenever two stations share an adjacent $\pm1$ constraint, they also share a co-channel constraint \autocite[Appendix~K, \S~2.2]{fcc-application-procedures}.
		In general, the physics of radio broadcasting allows us to assume that \(\mathcal{O}\) is a set of consecutive integers symmetric about zero: \(\mathcal{O} = \{-\max{\mathcal{O}}, -\max{\mathcal{O}} + 1, \ldots, \max{\mathcal{O}}\}\) for some non-negative integer \(\max{\mathcal{O}}\).
		The FCC's data have $\max{\mathcal{O}} = 2$.
		See section~\ref{decomposition} for more information about the FCC's data.

	\subsection{Feasibility and optimality}\label{feasibility-and-optimality}

		For a given pair of interference and domain constraints \((I, D)\) with
		corresponding station and channel sets \(S\) and \(C\), define the \introduce{feasible set} \(\BRPF(I, D)\) of the problem instance $(I, D)$ to be the set of assignments \(f : S \to C\) that are feasible with respect to \(I\) and \(D\) and complete with respect to \(S\).
		The \introduce{feasibility} variant of the BRP is to return zero if $\BRPF(I, D)$ is empty and one if it is non-empty.

		Unfortunately, consistency with both the FCC's terminology and common computer-science jargon requires our reusing \textit{feasible} in two contexts.
		A BRP instance is feasible when $\BRPF(I, D) \ne \emptyset$.
		A specific assignment is feasible (and thus possibly an element of \(\BRPF(I, D)\)) if it satisfies the list of constraints in Section~\ref{sec:constraints}.

		Given an assignment \(f : S' \to C\) for some $S' \subseteq S$, define the \introduce{cost} of an assignment to be the maximum channel $\|f\| = \max_{s \in S'}f(s)$ that $f$ assigns to any station.
		The \introduce{optimization} version of the BRP is to return
		\[
			\BRP(I, D) = \operatorname*{arg\,min}_{f \in \BRPF(I, D)}\|f\| = \operatorname*{arg\,min}_{f \in \BRPF(I, D)}\max_{s \in S}f(s).
		\]

		Minimizing the maximum channel is straightforward to describe mathematically and to interpret economically: it gives the maximum total amount of spectrum that the Incentive Auction could free up for mobile-network operators.
		Other objective functions, such as those the FCC defines \autocite{fcc-incentive-auction-comment-pn}, are also useful, but we just focus on $\|f\|$ here.

	\subsection{Subproblems}\label{subproblems}

		We are often interested in subproblems formed from \(\BRPF(I, D)\) by restricting ourselves to a subset of stations or a subset of channels.
		Given an input pair of interference and domain constraints \((I, D)\), let \(S\) and \(C\) be the corresponding sets of stations and channels.

		When we write \(\BRPF^{T}(I, D)\) for some \(T \subseteq S\) we implicitly take \(T\) as the set of stations, ignoring domain and interference constraints referring to stations in \(S \setminus T\).
		In this case, we consider assignments complete if and only if they are defined for all stations in \(T\).
		Stations in \(S \setminus T\) are said to be \introduce{cleared}, whereas we are \introduce{repacking} stations in \(T\).
		More precisely, if \(T \subseteq S\), we define
		\[
			\BRPF^T(I, D) = \BRPF\Big(
				\big\{\!
					\left((s_1, c_1), (s_2, c_2)\right) \in I \mid s_1, s_2 \in T
				\big\},
				\{C_s \in D \mid s \in T\}
			\Big).
		\]

		When we write \(\BRPF_{c^*}(I, D)\) for some \(c^* \in C\) we implicitly replace every domain constraint \(C_s\) with a new domain constraint \(\{c \in C_s \mid c \le c^*\}\).
		\(c^*\) is called the \introduce{clearing target}.
		More precisely, if \(c^* \in C\), we define
		\[
			\BRPF_{c^*}(I, D) = \BRPF\Big(
				I,
				\big\{\{c \in C_s \mid c \le c^*\} \subseteq C_s \mid C_s \in D\big\}
			\Big).
		\]

		If we write \(\BRP^T(I, D)\) or \(\BRP_{c^*}(I, D)\), we mean the optimization problem restricted to \(\BRPF^T(I, D)\) or \(\BRPF_{c^*}(I, D)\) respectively.
		When we refer in the abstract to \introduce{subproblems} of the BRP, we mean one of these two types of restrictions.

	\subsection{Summary of notation}\label{summary-of-notation}

		A BRP instance's input data are the pair \((I, D)\) and all the other data are defined implicitly from them.

		\begin{table}[H]
			\caption{Summary of Notation} \label{tab:title}
			\begin{tabular}[h]{ || p{5.4em} | p{18em} | p{13.9em} || }
				\hhline{||=|=|=||}
						\textbf{Symbol}
					&
						\textbf{Description}
					&
						\textbf{Notes}
				\\
				\hhline{||=|=|=||}
						\(S\)
					&
						stations
					&
						\(S \subseteq \mathbb{N}\)
				\\
				\hline
						\(C\)
					&
						channels
					&
						\(C \subseteq \mathbb{N}\)
				\\
				\hline
						\(f\)
					&
						assignment
					&
						\(f : S' \to C\) for $S' \subseteq S$
				\\
				\hline
						\(f(s)\)
					&
						station \(s\)'s assignment
					&
						if \(f\) is feasible, then $f(s) \in C_s$
				\\
				\hline
						\(I\)
					&
						interference constraints
					&
						\(I \subseteq (S \times C)^2\)
				\\
				\hline
						\(\mathcal{O}\)
					&
						offsets; constraint on structure of \(I\)
					&
						\(\mathcal{O} = \{-k, \ldots, k\}\)
				\\
				\hline
						\(C_s\)
					&
						station \(s\)'s domain
					&
						\(C_s \subseteq C\)
				\\
				\hline
						\(D\)
					&
						domain constraints
					&
						\(D = \{C_s\}_{s \in S}\); \(C = \bigcup D\)
				\\
				\hline
						\((I, D)\)
					&
						constraints
					&
						input data for a BRP instance
				\\
				\hline
						\({\BRPF(I, D)}\)
					&
						set of complete, feasible assignments
					&
					feasible with respect to \(I\) and \(D\), complete with respect to \(S	\)
				\\
				\hline
						\(\|f\|\)
					&
						cost of assignment \(f\)
					&
						$\max_{s \in S'}f(s)$
				\\
				\hline
						\(\BRP(I, D)\)
					&
						minimum-cost, complete, feasible assignment
					&
						defined only if \(\BRPF(I, D) \ne \emptyset\)
				\\
				\hline
						\(\BRPF^T(I, D)\)
					&
						\(\BRPF(I, D)\) but only with stations \(T\)
					&
						pack $T$, clear $S \setminus T$
				\\
				\hline
						\(\BRPF_{c^*}(I, D)\)
					&
						\(\BRPF(I, D)\) but only with channels~${\le c^*}$
					&
						\(c^*\) is the clearing target
				\\
				\hhline{||=|=|=||}
			\end{tabular}
		\end{table}

\section{The BRP as a graph problem}\label{brp-as-a-graph-problem}

	\subsection{The BRP's constraint and interference graphs}\label{sec:brp-graph-defs}

		To facilitate translating the BRP into more familiar graph problems, we define the interference and constraint graphs of a BRP instance and consider some of their subgraphs.
		For the remainder of this section, let \((I, D)\) be constraints defining a BRP instance, and let \(S\) and \(C\) be their corresponding sets of stations and channels.

		We want to work with undirected graphs, which Lemma~\ref{thm:undirectedness} below allows.
		It says that the order of interference constraints does not matter.
		We should just think of interference constraints as saying that two stations cannot be on certain channels at the same time.

		\begin{lemma}[Undirectedness]\label{thm:undirectedness}
			Suppose \(i = ((s_1, c_1), (s_2, c_2)) \in I\) and $i'=((s_2, c_2), (s_1, c_1))$.
			Then \[
				\BRPF(I, D)
				= \BRPF\left(I \cup \{i'\}, D\right)
				= \BRPF\left((I \setminus \{i\}) \cup \{i'\}, D\right).
			\]
		\end{lemma}

		\begin{proof}
			Apply the contrapositive to the definition an interference constraint:
			\begin{align*}
				\left(f(s_1) = c_1 \implies f(s_2) \ne c_2\right)
				&\iff \left(f(s_2) = c_2 \implies f(s_1) \ne c_1\right)
				\qedhere
			\end{align*}
		\end{proof}

		From now on, we will not worry about the order of the station-channel pairs in an interference constraint.
		Instead of thinking of \(i = ((s_1, c_1), (s_2, c_2)) \in I\) as an ordered pair, we will think of $i = \left\{(s_1, c_1), (s_2, c_2)\right\} \in I \subseteq [S \times C]^2$, where the notation $[Z]^k$ for a set $Z$ means all the subsets of $Z$ of size $k$.\footnotemark

		\footnotetext{
			Why didn't we just define $I$ this way to begin with?
			The FCC's constraint files and documentation are written as though order does matter.
			Lemma \ref{thm:undirectedness} permits us to interpret the files more flexibly.
		}

		Following \autocite{Frechette:2016it}, define the \introduce{interference graph} of \(\BRPF(I, D)\) as the simple graph whose vertex set is \(S\) and which has an edge between two stations if the stations share any interference constraint in \(I\).
		The interference graph $H$ superimposes all the interference constraints but carries none of the domain or at-most constraints.

		The constraint graph $G$ incorporates information about all the constraints all at once.
		Define the \introduce{constraint graph} of \(\BRPF(I, D)\) as the simple graph \(G\) in the following way: The vertices are station-channel pairs where stations are paired only with channels in the station's domain, so that
		\[
			V(G) = \left\{(s, c) \in S \times C \mid c \in C_s\right\}.
		\]
		The set of edges comes from the interference and at-most constraints.
		Two station-channel pairs in $V(G)$ for different stations share an edge if they share an interference constraint: \(\left(I \cap [V(G)]^2\right) \subseteq E(G)\).
		Two station-channel pairs in $V(G)$ for the same station always share an edge because of the at-most constraints: $A \subseteq E(G)$ where
		\begin{equation}\label{eq:atmost-edges}
			A = \left\{\left\{(s, c_1), (s, c_2)\right\} \in [S \times C]^2 \mid c_1, c_2 \in C_s \text{ and } c_1 \ne c_2\right\}.
		\end{equation}
		Together these are all of the edges:
		\[
			E(G) = \left(I \cap [V(G)]^2\right) \cup A.
		\]

		\begin{theorem}\label{thm:polynomial-time}
			The time complexity of  constructing a constraint graph \(G\) for \(\BRPF(I, D)\) from input \((I, D)\) is bounded by \(O(|I| + |D||\bigcup{D}|^2)\).
			The names of the stations and the names of the channels do not enter into this bound, so the time complexity is polynomial in the size of the input $(I, D)$ under a bit model of computation.
		\end{theorem}
		\begin{proof}
			Construct \(V(G)\) by iterating through each element of each domain constraint \(C_s\), which means looking at \(\sum_{s \in S}|C_s| \le |S||C|\) numbers.
			Construct \(E(G)\) by removing elements of \(I\) not in \(V(G)\), which can be done in \(O(|I|)\) steps, and by forming the at-most constraints \(A\), which can be done in \(O\left(\sum_{s \in S}\binom{|C_s|}{2}\right) \subseteq O(|S||C|^2)\) steps.
			All together, constructing \(G\) requires \(O(|S||C| + |I| + |S||C|^2) \subseteq O(|I| + |S||C|^2)\) steps.
			Every station has exactly one domain constraint, so \(|S| = |D|\). Finally, \(C = \bigcup{D}\).
		\end{proof}

		The size of the interference constraints are also bounded:
		\[
			|I| \le |[S \times C]^2| = |D|^2 |\bigcup D|^2.
		\]
		In the worst case, constructing the constraint graph takes $O(|D|^2 |\bigcup D|^2)$ steps.

		\begin{theorem}[Subproblems] \label{thm:subproblems}
			Consider \(\BRPF(I, D)\) with constraint graph G and a clearing-target subproblem \(\BRPF_{c^*}(I, D)\) or cleared-stations subproblem \(\BRPF^T(I, D)\) with constraint graph \(G'\).
			\(G'\) is an induced subgraph of \(G\).
		\end{theorem}

		\begin{proof}
			For the subproblem \(\BRPF_{c^*}(I, D)\), we create \(G'\) by removing from \(G\) any vertex corresponding to a channel greater than \(c^*\) but keep all edges not adjacent to the removed vertices.
			For the subproblem \(\BRPF^T(I, D)\), we remove any vertex corresponding to a station in \(S \setminus T\), but keep all edges not adjacent to the removed vertices.
		\end{proof}

		Later on, we will also consider two subgraphs of \(G\), which we will call \(G_I\) and \(G_I'\).
		\(G_I\) has just the interference constraints' edges: \(G_I = G - A\), where $A$ is the set of at-most constraints from Equation~\eqref{eq:atmost-edges}.
		In Section~\ref{component-partitioning}, we will define $G_I'$ so that \(G_I \subseteq G_I' \subseteq G\) and $G_I'$ has the same components as $G_I$.

		Notice that the at-most constraints form cliques in \(G\) from the vertices corresponding to each station.
		This gives us an easy lower bound on the clique number $\omega(G) \ge \max_{s \in S}{|C_s|}$.

	\subsection{Reduction to the maximum independent set problem}\label{reduction-to-mis}
		We can characterize $\BRPF(I,D)$ using its constraint graph \(G\) by identifying sets of vertices on \(G\) with assignments in $\BRPF(I,D)$.
		In set theory, a function $f : \mathcal{D} \to \mathcal{C}$ is defined to be a set of ordered pairs like $(x, y)$ where $x$ is in the domain $\mathcal{D}$ and $y$ is in the co-domain $\mathcal{C}$.
		We write this as $f \subseteq \mathcal{D} \times \mathcal{C}$.
		However, for $f$ to be a function and not just a relation, we require that for all \(x \in \mathcal{D}\) there exists at most one \(y \in \mathcal{C}\) such that $(x, y) \in f$, so we can write $f(x) = y$.
		Earlier, we defined assignments to be functions $f : S \to C$ and the vertex set \(V(G)\) of the constraint graph \(G\) to be a subset of \(S \times C\).
		That is, both \(f\) and \(V(G)\) are sets of station-channel pairs \((s, c)\) for \(s \in S\) and \(c \in C\).
		This means that we might be able to find subsets of vertices \(U \subseteq V(G)\) such that \(U\) is a function and thus an assignment.
		These subsets turn out to be $G$'s \introduce{independent sets} of vertices, sets in which none of the vertices are adjacent.

		In light of the fuzzy distinction between functions and sets of ordered pairs, we will think of the BRP constraints as being defined not specifically for assignments but for arbitrary subsets of \(S \times C\).
		We will use ordered pair notation rather than function notation for parts of the proof below.

		\begin{theorem}\label{thm:reduction-to-mis}
			Let \((I, D)\) be constraints, \(G\) be the constraint graph of \(\BRPF(I, D)\), $S$ and $C$ be the station and channel sets of $(I, D)$, and \(U \subseteq V(G)\).
			Let \(A\) denote the at-most constraint edges as defined in Equation~\eqref{eq:atmost-edges}.
			Then all the following hold:

			\begin{enumerate}[label=\thetheorem(\alph*)]
				\item \label{thm:reduction-to-mis-domain}
					\(U\) satisfies the domain constraints \(D\).
				\item \label{thm:reduction-to-mis-interference}
					\(U\) satisfies the interference constraints if and only if none of
					the vertices of \(U\) share an edge in \(I\).
				\item \label{thm:reduction-to-mis-assignment}
					\(U\) is an assignment if and only if none of the vertices of \(U\)
					share an edge in \(A\).
				\item \label{thm:reduction-to-mis-at-most}
					\(U\) satisfies the at-most constraints if and only if none of the
					vertices of \(U\) share an edge in \(A\).
				\item \label{thm:reduction-to-mis-feasible}
					\(U\) is a feasible assignment if and only if \(U\) is independent in
					\(G\).
				\item \label{thm:reduction-to-mis-complete-feasible}
					\(U\) is a complete, feasible assignment if and only if \(U\) is
					independent and \(|U| = |S|\).
			\end{enumerate}
		\end{theorem}
		\begin{proof}
			Keep in mind that $G$'s edges are either interference constraints or at-most constraints: $E(G) = (I \cup A) \cap [V(G)]^2$.

			\begin{enumerate}[label=\thetheorem(\alph*)]
				\item
					Satisfaction of the domain constraints comes directly from the definition of \(G\)'s vertex set.
				\item
					The interference constraints are defined by the proposition \((\forall \{(s_1, c_1), (s_2, c_2)\} \in I) ((s_1, c_1) \in U \implies (s_2, c_2) \not\in U)\), which is equivalent to the proposition
					\((\forall (s_1, c_1) \in U)(\forall (s_2, c_2) \in U) (\{(s_1, c_1), (s_2, c_2)\} \not\in I)\), which is equivalent to the proposition that none of the vertices of \(U\) share an edge in \(I\).
				\item
					Let \(S'\) be the set of stations corresponding to vertices in \(U\), i.e., \(S' = \{s \in S \mid \exists c \in C : (s, c) \in U\}\).
					Since \(U\) is already a subset of \(S' \times C\), \(U\) is an assignment if and only if \(U\) is a function, which is defined by the proposition \((\forall s \in S')((s, c_1), (s, c_2) \in U \implies c_1 = c_2)\).
					This is equivalent to the proposition \((\forall s \in S')(c_1 \ne c_2 \implies (s, c_1) \not\in U \lor (s, c_2) \not\in U)\).
					By definition of \(A\), this proposition is equivalent to no two vertices of \(U\) sharing an edge in \(A\).
				\item
					By definition, \ref{thm:reduction-to-mis-at-most} is equivalent to \ref{thm:reduction-to-mis-assignment}.
				\item
					\(U\) is a feasible assignment if and only if \(U\) is an assignment that satisfies the domain, interference, and at-most constraints.
					By \ref{thm:reduction-to-mis-domain}--\ref{thm:reduction-to-mis-at-most}, this occurs if and only if no two vertices of \(U\) share an edge in \(I\) or in \(A\).
					Since all the edges of \(G\) are either in \(I\) or \(A\), \(U\) is a feasible assignment if and only if no two vertices of \(U\) share an edge, which is to say, \(U\) is independent.
				\item
					By \ref{thm:reduction-to-mis-feasible}, $U$ is feasible if and only if $U$ is independent.
					A complete, feasible assignment requires exactly one station-channel pair per station.
					The at-most edges $A$ prevent $U$ from containing more than one station-channel pair per station.
					Thus $U$ is complete if and only if the vertices of $U$ share no edges of $A$ and $|U| = |S|$.
					\qedhere
			\end{enumerate}
		\end{proof}

		Let \(\alpha(G)\) be the \introduce{independence number}, the maximum cardinality of an independent set of vertices in a simple graph \(G\).
		An independent set in \(G\) is a \introduce{maximum independent set} (\introduce{MIS}) in \(G\) if its cardinality is \(\alpha(G)\).

		\begin{corollary}\label{thm:reduction-cor-alpha}
			Assume the same hypotheses as in Theorem~\ref{thm:reduction-to-mis}.
			\(\alpha(G) \le |S|\), with equality if and only if \(\BRPF(I, D)\) is non-empty.
			If \(\alpha(G) = |S|\), then \(\BRPF(I, D)\) is the set of maximum independent sets of vertices in \(G\).
		\end{corollary}

		\begin{proof}
			Let \(U\) be a maximum independent set.
			Then by Theorem~\ref{thm:reduction-to-mis-feasible}, \(U\) is an assignment, and the at-most constraints (which all assignments satisfy) dictate that \(|U| \le |S|\).
			The remaining statements follow directly from this and Theorem~\ref{thm:reduction-to-mis-complete-feasible}.
		\end{proof}

		The \introduce{maximum independent set problem} (\introduce{MISP}) for a graph $G$ has the optimization variant of finding a MIS of $G$, evaluation variant of calculating $\alpha(G)$, and recognition variant of deciding whether $\alpha(G) \ge k$ for a given $k \in \N$.

		\begin{corollary}\label{thm:brp-feas-to-misp}
			The feasibility variant of the BRP has a polynomial-time reduction to the MISP.
		\end{corollary}

		\begin{proof}
			Given constraints \((I, D)\), let \(G\) be the constraint graph of \(\BRPF(I, D)\), and let \(S\) be the corresponding set of stations.
			Solve the MISP on \(G\) and return the maximum independent set \(U\).
			By Corollary~\ref{thm:reduction-cor-alpha}, \(\BRPF(I, D)\) is non-empty if and only if \(|U| = |S|\).
			By Theorem~\ref{thm:polynomial-time}, this reduction takes polynomial time.
		\end{proof}

	\subsection{Optimality in the BRP}\label{optimality-of-brp}

		Recall that optimality in the BRP is obtained when we find the solution \(f \in \BRPF(I, D)\) that minimizes \(\|f\|\), the maximum channel assigned to a station.
		We want to convert this optimization problem into the decision problem MIS.
		We do so by creating the \introduce{recognition} problem: Given \((I, D)\) and a real number \(c^*\), return whether \(\|\BRP(I,D)\| \le c^*\).
		We know that \(\|\BRP(I,D)\| \in C\), so we may use binary search to determine which element of \(C\) is the optimal cost.
		This requires \(O(\log|C|)\) calls to the recognition problem \autocite[\S~11.8, p.~517-518]{Bertsimas:1997wc}.
		By definition, this problem is identical to the feasibility problem of determining whether the clearing-target subproblem's feasible set \(\BRPF_{c^*}(I, D)\) is non-empty.
		We may then turn to Corollary \ref{thm:brp-feas-to-misp} to solve this using the MISP.
		We summarize these considerations in the following corollary.

		\begin{corollary}\label{thm:brp-opt-to-misp}
			The optimization variant of the BRP has a polynomial-time reduction to the MISP.
		\end{corollary}

	\subsection{Complexity class}\label{complexity-class}

		The following theorem shows that the BRP turns out to be equivalent to list coloring after all because list coloring is $\mathcal{NP}$-complete \autocite{Kratochvil:1994ht}.
		Recall that $\mathcal{NP}$-completeness applies to the decision/recognition variant of problems.

		\begin{theorem} \label{thm:np-complete}
			The BRP is $\mathcal{NP}$-complete.
		\end{theorem}

		\begin{proof}
			Corollary~\ref{thm:brp-feas-to-misp} gives a polynomial time reduction of the BRP to the maximum independent set problem, which is \(\mathcal{NP}\)-complete \autocite[p.~1102]{Cormen:2009uw}.
			This shows that the BRP is in \(\mathcal{NP}\) \autocite[\S~11.8, p.~517-518]{Bertsimas:1997wc}.
			We will show that the BRP is \(\mathcal{NP}\)-hard, and thereby $\mathcal{NP}$-complete, by reducing graph \(k\)-coloring to the BRP, since graph \(k\)-coloring is \(\mathcal{NP}\)-complete \autocite{Karp:1972tu}.

			We define an instance of $k$-coloring and of the BRP as follows.
			Let
			\begin{align*}
				Q &\text{ be a simple graph},             & k &\in \N \text{ such that } k \ge 3, \\
				      C &= \{1, \ldots, k\},              & S &= V(Q), \\
				    C_s &= C \text{ for all } s \in S,    & D &= \{C_s\}_{s\in S}, \text{ and} \\
				      I = \rlap{$\displaystyle \left\{\{(s_1,c_1), (s_2,c_2)\} \in [S \times C]^2 \mid s_1s_2\in E(Q) \text{ and } c_1 = c_2\right\}$.} &&&
			\end{align*}
			We will show that \(\BRPF(I, D)\) is non-empty if and only if $Q$ has a proper \(k\)-coloring \(f\).

			Suppose \(\BRPF(I, D)\) is non-empty.
			Then it contains a complete, feasible assignment \(f\) and thus a coloring of $Q$.
			By the domain constraints $D$, \(f\) uses only \(k\) different values, so \(f\) is a \(k\)-coloring of \(Q\).
			The interference constraints \(I\) prevent $f$ from assigning the same channel to stations that are neighboring vertices in \(Q\), so $f$ is a proper $k$-coloring of $Q$.

			Conversely, suppose there is a proper \(k\)-coloring \(f'\) of \(Q\).
			Create \(f\) from \(f'\) by relabeling the range of \(f\) with the elements of \(C\).
			Then \(f\) is a proper \(k\)-coloring of \(Q\) such that \(f(s) \in C_s\) for every \(s \in S\).
			Consequently, \(f\) is an assignment that obeys the domain constraints.
			Finally, \(f\) obeys the interference constraints because if \(s_1s_2 \in E(Q)\) then \(f\)'s being a proper coloring means that \(f(s_1) \ne f(s_2)\).
			Therefore, $f$ respects the interference constraints in \(I\), as it is defined above, and we have \(f \in \BRPF(I, D)\).
		\end{proof}

\section{Decomposition of the constraint graph}\label{decomposition}

	Even though the MIS problem may not have a polynomial time algorithm, the instances of the BRP we are interested in are tractable because of the sparsity of \(G\).
	Prerelease versions of TVRepack, which implements none of the heuristics or caching schemes SATFC does, can solve typical inputs to optimality in half a minute.
	SATFC version 2.0 can solve nearly all the feasibility subproblems \textcite[Figure~5]{Frechette:2016it} tested, each in less than a second.

	\textcite{Frechette:2016it} discuss effective methods for decomposing the constraint graph to help their satisfiability-programming and local-search algorithms run faster.
	We explore how to apply two of their recommendations to the MIS formulation:
	\begin{enumerate}
		\item Underconstrained stations, and
		\item Component partitioning.
	\end{enumerate}
	We can actually take more advantage of the sparsity of $G$ than \textcite{Frechette:2016it} did, so the presentation below is novel.

	The constraint-graph description of the BRP also allows us to quantify exactly what these heuristics do to help us solve the BRP because we can describe them using standard graph-theoretic parameters of the constraint graph.
	We take our problem instance $(I, D)$ from the FCC's November 2015 constraint files \href{http://data.fcc.gov/download/incentive-auctions/Constraint_Files/}{\texttt{Domain.csv} and \texttt{Interference\_Paired.csv}}.
	These are the last constraint files the FCC published before the auction began \autocite{fcc-constraint-files}.

	The FCC's data contain 2,990 stations with offsets $\mathcal{O} = \{-2, \ldots, 2\}$.
	Most stations have co-channel and adjacent $\pm{1}$ interference constraints, but the adjacent $\pm 2$ constraints apply only to the 793 Canadian TV stations, whose owners do not get to bid in the auction.
	Canada agreed to repack their TV stations jointly with the US, which moves the constraints preventing interference with Canada's stations from US stations' domain constraints to their interference constraints \autocite[Appendix~K, \S~2]{fcc-application-procedures}.
	(Mexican stations' channels have been removed from nearby US stations' domains to prevent interference with Mexican broadcasters, who are not in the constraint files themselves \autocite[Appendix~K, \S~3]{fcc-application-procedures}.)
	The resulting constraint graph $G$ has 101,868 vertices and 4,713,968 edges, 0.09 percent of the total possible edges.
	The minimum degree is $\delta(G) = 4$ and the maximum degree is $\Delta(G) = 229$.

	\subsection{Underconstrained stations}\label{underconstrained-stations}

		\textcite{Frechette:2016it} write that a station is \introduce{underconstrained}
		\begin{quote}
			when there exist stations for which, regardless of how every other station is assigned, there always exists some channel into which they can be packed.
			Verifying this property exactly costs more time than it saves; instead, we check it via the sound but incomplete heuristic of comparing a station's available channels to its number of neighboring stations \autocite[p.~5]{Frechette:2016it}.
		\end{quote}

		The FCC actually solves exactly for the full set of underconstrained stations \autocite[Appendix~J]{fcc-application-procedures}, but Fr{\'e}chette's heuristic still simplifies  the feasibility problem considerably.
		The number of stations in the FCC's data satisfying this heuristic is 1,424, nearly half of all stations.
		The constraint graph after removing the underconstrained stations has 54,406 vertices and 2,388,160 edges, each of which are about half the corresponding values for the constraint graph of the full problem.

		We now prove the authors' assertion that their heuristic can identify underconstrained stations.
		For a graph \(G\) and one of its vertices \(v\), denote \(v\)'s degree in \(G\) as
		\(d_G(v)\) and \(v\)'s neighborhood in \(G\) as \(N_G(v)\).

		\begin{lemma}\label{thm:underconstrained}
			Let \(H\) be the interference graph of \(\BRPF(I, D)\) for constraints \((I, D)\) with station set \(S\).
			A station \(s \in S\) is underconstrained if, for \(C_s \in D\), we have \(|C_s| > d_H(s)\).
		\end{lemma}

		\begin{proof}
			Given a station \(s \in S\) and a feasible assignment \(f : S \setminus \{s\} \to C\) we may extend \(f\) to a complete, feasible assignment \(F : S \to C\) by setting \(F(t) = f(t)\) for \(t \in S \setminus \{s\}\) and take
			\(F(s) \in C_s \setminus \{f(t) \mid t \in N_H(s)\}\), which is non-empty because \(|C_s| > d_H(s) = |N_H(s)|\).
		\end{proof}

		Once a preprocessor identifies a set $T^\complement$ of underconstrained stations, we calculate \(f = BRPF^{T}(I, D)\).
		We then extend this assignment $f$ to all of \(S\) by setting \(F(s) = f(s)\) for all \(s \in T^\complement\) and setting \(F(t) = \min\left(C_s \setminus \{f(s) \mid s \in N_H(t)\}\right)\) for each \(t \in T\).
		In other words, we solve all the stations that are not underconstrained the hard way, then greedily assign the underconstrained stations the least channel not used by neighbors in the interference graph $H$.
		While this obtains a feasible solution, it may not be optimal for minimizing the maximum channel we assign to all stations, so this technique is for the feasibility form of the BRP.
		As usual, we may combine iterations of the feasibility-form of the BRP with lower and lower clearing targets to find an optimum solution.

	\subsection{Component partitioning}\label{component-partitioning}

		If we solve the MISP on each of a graph \(G\)'s components, any combination of the MISs found on each of the components gives an independent set on \(G\).
		When $G$ is a BRP constraint graph, the vertices representing a given station form a clique, so each station must show up in exactly one component of $G$.

		In practice, the subgraph $G_I$, whose edges are just the interference constraints, has a lot of components compared to G.
		Using the FCC's November 2015 data, $G$ has 171 components and no isolated vertices.
		With just the interference constraints, $G_I$ has 1,205 nontrivial components and 1,424 isolated vertices.
		(Notice that there are 1,424 underconstrained stations by the analysis of the previous section.)
		The largest three components of $G$ have 89,401; 1,934; and 1,258 vertices respectively.
		The largest three components of $G_I$ have 51,665; 15,357; and 6,200 vertices respectively.

		According to \textcite{Frechette:2016it} this difference in connectedness is partially because the set of channels \(C\)
		is partitioned into three equivalence classes across which there are no
		interference constraints:

		\begin{quote}
			LVHF (channels 1--6), HVHF (channels 7--13), and UHF (channels 14--\(\bar{c}\), excepting 37), where \(\bar{c} \le \) 51 is the largest available UHF channel set by the auction's clearing target, and 37 is never available.
			No interference constraints span the three equivalence classes of channels, giving us a straightforward way of decomposing the problem \autocite[p.~2]{Frechette:2016it}.
		\end{quote}

		This view is incomplete.
		Channels only interfere when they are within one or two (or $\max{\mathcal{O}}$) of each other.
		$G_I$ reflects all of this structure.
		Using the structure of $G_I$ means solving MIS on hundreds of tiny components (fewer than a dozen vertices), which can be much faster than solving on a small number of very large components because practical MIS algorithms can be worst-case exponential.

		However, large independent sets found in components of $G_I$, lacking the at-most-constraints' edges $A$ (Eq.~\ref{eq:atmost-edges}), may violate the at-most constraints.
		We mitigate this by adding back into \(G_I\) at-most constraints that do not cross components of \(G_I\).
		Call this new graph $G_I'$.
		In the FCC's data, $G_I'$ has 1,098,268 more edges than $G_I$ for a total of 3,673,734 edges.
		We find a maximum independent set $U_i$ on each component $i$ of \(G_I'\).
		While $U = \bigcup_i U_i$ is an independent set of $G_I'$, $U$ may violate the at-most constraints in $G$ that crossed components of \(G_I\).
		$U$ satisfies all other constraints, and each $U_i$ contains at most one vertex per station.
		We may choose one arbitrary (alternatively, channel-minimum) vertex in $U$ per station to form a maximum independent set of vertices in all of \(G\).

		\textcite[p.~5]{Frechette:2016it} report that, in their experiments on hundreds of thousands of instances of the BRP, ``runtimes were almost always dominated by the cost of solving the largest component.''
		They concluded that the simplicity of solving the components serially outweighed any possible speed gains from solving them in parallel.
		Moreover, in the scheme we are considering here, as soon as any vertex for a station is chosen from a component, vertices corresponding to that station may be ignored in all subsequent components.
		We exploit this by solving the components of $G_I'$ in ascending order of size.

		To determine infeasibility before we have examined all components of \(G_I'\), we keep track of the vertices explored so far corresponding to each station.
		We maintain $X_i(s)$ equal to the number of vertices corresponding to station $s$ seen in component $i$ of $G_I'$ and all previous components.
		If $X_i(s) = |C_s|$ for some $s \in S$ but we have not yet found a channel assignment for $s$ after exploring component $i$, we know that \(s\) will be assigned no channel.
		The algorithm cannot return a complete assignment, so this instance of the BRP is infeasible.

\section{Zero-one integer programming and satisfiability formulations}\label{sec:zoip}

	Even after decomposing a BRP, we may still have to solve some subproblems the hard way.
	Zero-one integer programming (\introduce{ZOIP}) makes it easy to solve the constraint graph's MIS problem in commonly available software.
	Satisfiability (\introduce{SAT}) solvers are also commonly available, and can be quite powerful for feasibility testing.
	In this section, we don't worry about any of the preprocessing described above or how to combine MIS solutions on components into a MIS of the whole constraint graph.
	The procedures above can work with any black-box MIS solver.
	We just use ZOIP and SAT here as two such possible MIS solvers.

	Let $G$ be a connected subgraph of the constraint graph of some problem instance $\BRPF(I, D)$, and let $S$ be the corresponding set of stations.
	Take $x_v$ as the binary decision variable indicating whether to include vertex $v$ in a MIS.
	The ZOIP to find a maximum independent set of $G$ is
	\begin{align}\label{eq:miszoip}
		\text{maximize }  & \sum_{v \in V(G)} x_v &&\\
		\notag
		\text{such that } & x_u + x_v \le 1 &\text{ for all } uv &\in E(G) \\
		\notag
		                  & x_v \in \{0, 1\} &\text{ for all } v &\in V(G).
	\end{align}

	The SAT program for solving the BRP is similar.
	Take $t_{(s,c)}$ as the SAT variable (a true/false proposition) indicating whether to include $(s,c) \in V(G)$ in an independent set, where $s \in S$ is a station and $c \in C_s$ is a channel in station $s$'s domain.
	A SAT formula in conjunctive normal form is the conjunction of the following disjunctions.
	\begin{align}\label{eq:missat}
		\neg t_{(s_1, c_1)} \vee \neg t_{(s_2, c_2)} &\text{ for all } \left\{(s_1, c_1), (s_2, c_2)\right\} \in E(G) \\
		\notag
		\bigvee_{c \in C_s} t_{(s, c)} &\text{ for all } s \in S
	\end{align}
	The first clause encodes the adjacency relation in $G$, and the second clause requires the selection of at least one channel for each station.
	By Theorem~\ref{thm:reduction-to-mis-complete-feasible}, Equation~\eqref{eq:missat} is feasible if and only if $\BRPF(I, D)$ is feasible.

	\begin{theorem}
		The ZOIP defined by \eqref{eq:miszoip} returns a maximum independent set of graph $G$.
	\end{theorem}
	\begin{proof}
		Let $U^*$ be a maximum independent set of $G$.
		($U^*$ is non-empty because a set containing a single vertex is independent.)
		Define for each $v \in V(G)$,
		\[
			x_v = \mathbb{1}_{U^*}(v) =
			\begin{cases}
				1 & v \in U^* \\
				0 & v \not \in U^*.
			\end{cases}
		\]
		$x_v \in \{0, 1\}$ and since $U^*$ is independent, $x_u + x_v \le 1$ for all $uv \in E(G)$.
		Thus $\{x_v\}_{v\in V(G)}$ is a feasible solution to \eqref{eq:miszoip}.
		The cost of this solution is
		\[
			\sum_{v \in V(G)} x_v = |U^*| = \alpha(G).
		\]

		Now let $\{x_v\}_{v\in V(G)}$ be a feasible solution to \eqref{eq:miszoip}.
		Define $U = \{v \in V(G) \mid x_v = 1\}$.
		No two decision variables can take the value one simultaneously if their corresponding vertices share an edge, so $U$ is an independent set.
		The cost of the $x_v$s is
		\[
			\sum_{v \in V(G)} x_v = |U| \le \alpha(G).
			\qedhere
		\]
	\end{proof}

	For a ZOIP solver that uses Gomory cutting planes, certain additional constraints can improve performance, including the constraints $\sum_{v \in U}x_v \le 1$ for any clique $U$ in $G$, and $\sum_{v \in W}x_v \le \frac{|W| - 1}{2}$ for any odd-length cycle $W$ in $G$ \autocite[Example~11.3]{Bertsimas:1997wc}.
	Finding every clique and odd cycle of $G$ can be difficult.
	In fact, finding maximum cliques is just the same as finding MISs in the complement of the graph.
	However, in the BRP we already know how to find some large cliques because the at-most constraints induce one per station.
	We may then augment \eqref{eq:miszoip} by adding explicit at-most constraints:
	\begin{align}\label{eq:miszoip-cliques}
		\text{maximize }  & \sum_{v \in V(G)} x_v &&\\
		\notag
		\text{such that } & x_u + x_v \le 1 &\text{ for all } uv &\in E(G) \\
		\notag
		                  & \sum_{(s, c) \in V(G)} x_{(s,c)} \le 1 &\text{ for all } s &\in S \\
		\notag
		                  & x_v \in \{0, 1\} &\text{ for all } v &\in V(G).
	\end{align}

\section{Conclusion and future research}\label{sec:future-research}
	Starting with test cases the FCC derived from auction simulations, \textcite{Frechette:2016it} were able to solve over 99 percent of BRP instances for feasibility in under a second using SAT and local search solvers together with a clever caching scheme.
	This despite the BRP's $\mathcal{NP}$-completeness and the FCC's data's having millions of constraints.
	But the Incentive Auction is not over and the government's actual repacking of real broadcasters has not yet happened.
	In the foregoing discussion, we made a novel application of constraint graphs to the BRP.
	The BRP's $\mathcal{NP}$-completeness itself means that the constraint-graph technique can clarify a wide class of similar problems.
	Further research into the BRP may therefore improve solvers' speeds in this application and many others.

	Usage of ZOIP to solve the BRP may be good way forward.
	SATFC's authors' initial tests on ZOIP solvers CPLEX and Gurobi were disappointing but they ``did not investigate such alternatives in depth'' \autocite[p.~3]{Frechette:2016it}.
	Future research is required to determine whether the augmented problem decompositions we propose above would make ZOIP solvers competitive with SAT solvers, and whether linear-programming-specific techniques such as dynamic column generation can be used to improve performance.
	In particular, our own experiments have shown that copying data from TVRepack into SAT and ZOIP solvers takes considerably more time than solving on sparsely constrained problem instances.

	Breaking down the constraint graph into components and then finding MISs on $G_I'$ works if we solve the optimization problem as a sequence of feasibility problems.
	Ideally, we would ask the ZOIP solver for the MIS in a component that minimizes the global objective function.
	This would be easy if we wanted to find the MIS of the entire constraint graph $G$ in one call to the MIS solver:
	\begin{align} \label{eq:optimize-brp}
		\text{minimize }  & z &&\\
		\notag
		\text{such that } & x_u + x_v \le 1 &\text{ for all } uv &\in E(G) \\
		\notag
		                  & \sum_{v \in V(G)} x_v = |S| && \\
		\notag
		                  & cx_{(s,c)} \le z &\text{ for all } (s, c) &\in V(G) \\
		\notag
		                  & x_v \in \{0, 1\} &\text{ for all } v &\in V(G).
	\end{align}
	By Theorem~\ref{thm:reduction-to-mis-complete-feasible}, the first and second constraints in \eqref{eq:optimize-brp} force the ZOIP to be feasible only if $\BRPF(I, D)$ is non-empty, and the optimal solution is $\BRP(I, D)$.
	The trouble with decomposition arises because in the decomposed problem on $G_I'$, we don't know yet whether every station represented in a component will also be represented in a given MIS of that component.
	In other words, we don't know the independence number of each component, only of the constraint graph as a whole, where $\alpha(G) = |S|$ if and only if $\BRPF(I, D)$ is non-empty.
	More research is required to determine whether, once a MIS has been found in a given component, a solver can easily find other MISs in order to minimize the objective function.
	At the very least, it would be helpful to set up the constraints and objective function in the ZOIP so the solver didn't have to be called multiple times.
	One possibility is that once we have the cardinality of the MIS of a component, we solve the Lagrangian dual of the MIS problem's ZOIP to search for a MIS of the known size while minimizing the maximum channel corresponding to a vertex of the MIS found.

	The BRP is a tractable $\mathcal{NP}$-complete problem when we apply the constraint-graph technique and will become even more so as research progresses.
\bibliography{brp.bib}
\end{document}